\documentclass[conference]{IEEEtran}

\usepackage[utf8]{inputenc}
\usepackage[svgnames]{xcolor}
\usepackage{graphicx}
\usepackage{comment}
\usepackage{tikz}
\usetikzlibrary{math, shapes}
\usepackage{pgfplots}
\pgfplotsset{compat=1.9,legend style ={font=\footnotesize}, width=7cm, axis lines = left, yminorticks=false}
\usepackage{epstopdf}
\usepackage{subcaption}
\usepackage{cite}
\usepackage{amssymb,amsmath,amsfonts,amsthm,bm,mathtools,cuted}
\usepackage{tabularx}
\usepackage{booktabs}
\usepackage[ruled,vlined,linesnumbered]{algorithm2e}

\newcommand{\T}{^{\rm {T}}}
\newcommand{\E}[1]{\mathbb{E}\left\{ #1 \right\}}
\newcommand{\mc}[1]{\mathcal{#1}}
\newcommand{\mb}[1]{\mathbf{#1}}
\newcommand{\code}[1]{\texttt{#1}}

\DeclareMathOperator*{\argmax}{arg\,max}

\newtheorem{proposition}{Proposition}

\title{NOMA Power Minimization of Downlink Spectrum Slicing for eMBB and URLLC Users} 

\author{\IEEEauthorblockN{Fabio Saggese}
\IEEEauthorblockA{\textit{Dept. of Information Engineering} \\
\textit{University of Pisa}\\
Pisa, Italy \\
\texttt{fabio.saggese@phd.unipi.it}}
\and
\IEEEauthorblockN{Marco Moretti}
\IEEEauthorblockA{\textit{Dept. of Information Engineering} \\
\textit{University of Pisa}\\
Pisa, Italy \\
\code{marco.moretti@unipi.it}}
\and
\IEEEauthorblockN{Petar Popovski}
\IEEEauthorblockA{\textit{Dept. of Electronic Systems} \\
\textit{Aalborg University}\\
Aalborg, Denmark \\
\code{petarp@es.aau.dk}}
}
\date{}

\begin{document}
\maketitle
\begin{abstract}
    \emph{Spectrum slicing} of the shared radio resources is a critical task in 5G networks with heterogeneous services, through which each service gets performance guarantees. In this paper, we consider a setup in which a Base Station (BS) should serve two types of traffic in the downlink, enhanced mobile broadband (eMBB) and ultra-reliable low-latency communication (URLLC), respectively. Two resource allocation strategies are compared: non-orthogonal multiple access (NOMA) and orthogonal multiple access (OMA). 
    A framework for power minimization is presented, in which the BS knows the channel state information (CSI) of the eMBB users only. Nevertheless, due to the resource sharing, it is shown that this knowledge can be used also to the benefit of the URLLC users. The numerical results show that NOMA leads to a lower power consumption compared to OMA for every simulation parameter under test.
\end{abstract}

\section{Introduction}
\label{sec:intro}
5G technology is conceived as a connectivity platform that is capable to support flexibly a plethora of services with heterogeneous requirements. To address the complexity of such a vast connectivity space, 5G has opted to define three generic service types: enhanced mobile broadband (eMBB), massive machine-type communications (mMTC), and ultra-reliable low-latency communications (URLLC). These generic types should not be seen as exclusive to a device or a service; there will be composite services that use any combination of the aforementioned generic service types. 

In this context, \emph{spectrum slicing} is a fundamental task in a 5G Radio Access Network (RAN): how to utilize a given chunk of spectrum to serve users with heterogeneous requirements. This terminology is motivated by the more general concept of \emph{network slicing}: partitioning the physical network infrastructure in different end-to-end isolated virtual networks, i.e. \emph{slices}, each one able to support a different service requirement for different use cases~\cite{Zhang2017}. In that sense, we can also refer to spectrum slicing as RAN slicing. 

The general problem of resource allocation for network slicing is discussed in~\cite{Doro2019}.  Several prior works treated the problem of resource allocation algorithms to multiplex eMBB and URLLC services. In~\cite{Anand2020}, the joint resource allocation problem for eMBB-URLLC slicing is addressed employing different puncturing models. The authors of~\cite{Elsayed2019} and~\cite{Alsenwi2021} propose different deep reinforcement learning techniques to allocate the resource to the two services, employing orthogonal resources and pre-emption/puncturing, respectively. All these works consider the standard application of dynamic resources sharing between eMBB and URLLC, i.e. by means of puncturing or non-overlapping time/frequency resources~\cite{3gpp:access}.
However, in certain cases power domain non-orthogonal multiple access (NOMA) has shown the potential to outperform orthogonal multiple access (OMA)~\cite{Xu2015}. Spectrum slicing can be realized either by OMA or by NOMA, as introduced in the framework presented in~\cite{Popovski2018}\footnote{Note that the terms OMA and NOMA are commonly used to denote the sharing of wireless resources among services with the same type of requirement, e.g., rate maximization. In~\cite{Popovski2018} the terms H-OMA (heterogeneous OMA) and H-NOMA (heterogeneous NOMA) are used to emphasize the heterogeneous requirements; here we will use NOMA and OMA without any risk of confusion.}. The use of NOMA is also been investigated in~\cite{Kalor2019} for URLLC devices with different latency requirements.
In~\cite{Li2019}, a reinforcement learning algorithm decides whenever to use OMA or NOMA for dynamic multiplexing of eMBB-URLLC data streams, setting the transmission power based on the information about the channel gain. 

\begin{figure}[t!]
    \centering
    \includegraphics[width=8cm]{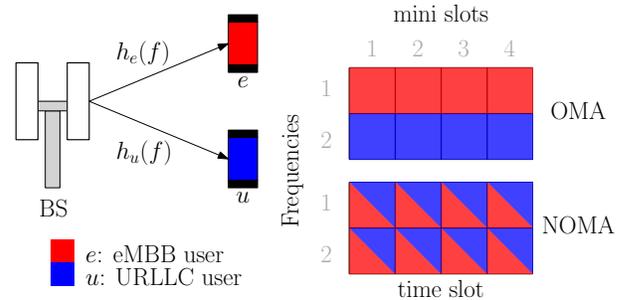}
    \caption{An elementary  example of the system model with available frequencies $F = 2$ and available mini-slots $M = 4$.}
    \label{fig:toy-example}
\end{figure}

In this paper, we employ NOMA to enforce spectrum slicing for downlink transmissions of eMBB and URLLC traffic to minimize the transmitted power. The complexity of the model arises due to the different characteristics of the two traffic types.  The eMBB service aims to maximize the throughput, while its latency requirements are negligible. 
Differently from this, the URLLC service demands mission-critical, reliable communication, where a hard latency constraint must be fulfilled. Accurate channel state information (CSI) is not available for URLLC traffic, due to its intermittent nature and the stringent latency constraints, while it can be acquired for the eMBB device. Accordingly, we assume that instantaneous CSI is known at the BS only for the eMBB devices, while for the URLLC devices only statistical CSI is available. The different levels of knowledge of CSI impact the resource allocation, but it is not obvious to draw clear conclusions. Interestingly, it still turns out that, in non-orthogonal schemes, the URLLC users can benefit from the fact that the BS has a CSI of the eMBB users.  

In particular, the problem we address here is exemplified in Fig.~\ref{fig:toy-example}, which shows two users, one for each service, transmitting over a limited set of resources. Since the methods discussed here can be easily generalized to more crowded scenarios, our work is preliminary to more general analysis. Even in this simple scenario, the allocation problem depends on various parameters, which greatly complicate the solution of the slicing problem.
Possessing only statistical CSI, the performances of the URLLC are formulated by means of the outage probability. Unfortunately, the closed-form expression of outage probability for transmissions on parallel channels is not known in the general case, while its approximations involve very complicated non-convex special functions~\cite{Bai2013, Coon2015, Li2020}. Accordingly, after discussing in-depth the slicing problem, we propose a novel power allocation algorithm based on a pragmatical approach to compute the outage probability for URLLC users.

Eventually, numerical results show that our NOMA approach outperforms the more conventional OMA.

\section{System Model}
\label{sec:model}
We consider a downlink communication scenario where the base station (BS) transmits eMBB and URLLC traffics. To capture the nature of the dynamics of the slicing problem, our study is focused on a simplified scenario with only two users: an eMBB user  $e$  and a URLLC user $u$. 

Following a 5G-like model, the available resources are organized in a  time and frequency grid. The time axis is organized in time slots, each divided into $M$ mini-slots collected in the set $\mc{M}$. Following the 5G NR standard~\cite{3gpp:access}, the eMBB traffic is allocated on a time slot basis and URLLC traffic on a mini-slot basis. 
We assume that channel coherence time is always greater than the duration of a single time slot.
In the frequency domain, there are $F$ orthogonal frequency resources (FR) collected in the set $\mc{F}$.
We further denote as $\mc{F}_e \subseteq \mc{F}$ and $\mc{M}_e \subseteq \mc{M}$ the sets of frequency and temporal resources reserved to user $e$, and with $\mc{F}_u \subseteq \mc{F}$ and $\mc{M}_u \subseteq \mc{M}$ the sets of resources reserved to user $u$. The cardinality of these sets is $F_e, M_e, F_u, M_u$, respectively.
While in a standard OMA configuration the resources are mutually exclusive, employing the NOMA paradigm the resource grid is shared between URLLC and eMBB, i.e.  $\mc{F}_u \times \mc{M}_u = \mc{F}_e \times \mc{M}_e$.

The system performance is computed as a function of the overall transmitted power. The eMBB user is modeled as transmitting with an overall spectral efficiency of $r_e$ [bit/s/Hz]  with the minimum possible transmitted power.
The requirements for the URLLC user are formulated in terms of latency and reliability. In particular, we assume that a packet with spectral efficiency $r_u$ [bit/Hz/s] must be delivered within $T_u$ seconds with an outage probability lower than $\epsilon_u$.
The URLLC latency constraints are expressed as a function of the \emph{edge delay}, i.e. the delay between the time at which the message arrives at the transmitter and the time at which the message is effectively transmitted, assuming that all other delay terms have already subtracted from $T_u$~\cite{She2017}. Without loss of generality, we define the tolerable latency in terms of the number of mini-slots. Hence, each packet should be correctly decoded by the receiver within latency $l^{\max}$.
Only the mean large-scale fading gain $\Gamma_u$ is known at the transmitter.

\subsection{Signal model}
\label{sec:signalmodel}
We denote as $n$ the number of symbols contained in a mini-slot. 
Due to the latency requirements, the length of each URLLC codeword is assumed to be equal to $n$, so each mini-slot contains a different codeword. On the other hand, a codeword of $e$ may span  multiple mini-slots, e.g. $L \le M$, assuming that its duration is exactly a multiple of $n$.
Let us consider a single resource $(t, f)$, i.e. a single mini-slot $t\in\mc{M}$, and a single FR $f \in \mc{F}$.
Therefore, the vectors of transmitted symbols in a resource are denoted as $\mb{s}_u(t, f) = [s_u(1, f), \dots, s_u(n, f)]\T$ for $u$, and $\mb{s}_e(t, f) = [s_e((t - 1)n + 1, f), \dots, s_e(t \, n, f)]\T$ for $e$. We assume that the symbol vectors chosen follows $\E{||\mb{s}_i(t, f)||^2} = 1$, $i\in\{e,u\}$, $\E{\mb{s}_e^{\rm H}(t, f) \mb{s}_u(t, f)} = 0$.
The BS transmits both eMBB and URLLC data streams using superposition coding. The transmitted signal in a single resource $(t,f)$ is:
\begin{equation}
\begin{aligned}
\mb{x}(t, f) &= \sqrt{P_e(t,f)} \mb{s}_e(t, f) + \sqrt{P_u(t,f)}  \mb{s}_u(t, f)
\end{aligned}
\end{equation}
where $P_e(t,f)$, $P_u(t,f)$ are the power coefficients used to transmit the symbols of $e$ and $u$ on resource $(t,f)$, respectively.
It is worth noting that this formalization implies the possibility of the transmission of the data stream of a single user $i \in \{u,e\}$ in an OMA fashion, by imposing the other coefficient equal to 0.

On the receiver side, we can model the signal received by user $i \in \{e,u\}$ as
\begin{equation}
\begin{aligned}
\mb{y}_i(t, f) = h_i(f) \mb{x}(t,f) + \mb{n}_i
\end{aligned}
\end{equation}
where $h_i(f)$, $i \in \{e,u\}$, is the the channel gain, and specifically $h_e(f)$ is assumed known, while $h_u(f) \sim \mc{CN}(0, \Gamma_u)$; $n_i \sim \mc{N}(0, \sigma_i^2 \mb{I}_{n})$, $i \in \{e,u\}$, is the noise at the receiver.
We denote the instantaneous normalized channel gain at the receiver as $i\in\{e,u\}$ as 
\begin{equation}
    \gamma_i(f) = \frac{|h_i(f)|^2}{\sigma_i^2},
\end{equation}
where it is worth noting that $\gamma_e(f)$ is a deterministic value and it is known, while $\gamma_u(f) \sim \text{Exp}(\rho_u)$, where $\rho_u = \E{\gamma_u(f)} = \Gamma_u/ \sigma_u^2$ is the normalized average channel gain.

Employing NOMA, each user may implement successive interference cancellation (SIC) to cancel the data stream of the other user and recover its data stream with virtually no error. However, the cancellation of a user data stream depends on the reception of the whole codeword, but the eMBB codewords may span over multiple mini-slots. Thus, if the URLLC user waits for the reception of an entire $e$ codeword, it risks violating its latency requirement. Therefore, we concentrate on a NOMA communication paradigm where SIC is employed only by the eMBB user, regardless of the quality of its channel.

\subsection{Mutual information}
Having considered a quasi-static fading channel, the channel dispersion is zero, and the description of the maximum available rate for short packets may be done in terms of outage capacity~\cite{Polyanskiy2010, Yang2014, Durisi2016}.
Accordingly, the mutual information of $u$ data stream at receiver $u$ is denoted as~\cite{Li2020}
\begin{equation} \label{eq:info:u}
    I_u = \sum_{t\in\mc{M}_u} \sum_{f\in\mc{F}_u} \log_2\left( 1 + \frac{\gamma_u(f) P_u(t, f)}{1 + \gamma_u(f) P_e(t, f)}\right),
\end{equation}
the mutual information of $u$ data stream at receiver $e$ is denoted as
\begin{equation} \label{eq:info:ue}
    I_{u,e} = \sum_{t\in\mc{M}_u} \sum_{f\in\mc{F}_u} \log_2\left( 1 + \frac{\gamma_e(f) P_u(t,f)}{1 + \gamma_e(f) P_e(t,f)}\right),
\end{equation}
and the mutual information of $e$ data stream at receiver $e$ after the SIC process is
\begin{equation} \label{eq:info:e}
    I_e = \sum_{t\in\mc{M}_e} \sum_{f\in\mc{F}_e} \log_2\left( 1 + \gamma_e(f) P_e(t,f)\right).
\end{equation}

\subsection{Outage events}
\label{sec:outages}
Following the mutual information formalization, we can define the outage events of both users.

\paragraph*{eMBB outages}
The data stream transmitted to $e$ is incorrectly decoded if: a) the SIC process is not successful, or b) the data stream of $e$ is erroneously decoded after the SIC.

a) SIC fails if $r_u > I_{u,e}$. The probability of this event, denoted as  $\mc{O}_{u,e}$, is
\begin{equation} \label{eq:outage:sic}
    P(\mc{O}_{u,e}) = \Pr\{ I_{u,e} < r_u \}.
\end{equation}

b) Assuming that SIC has been successful, the data stream of $e$ is wrongly decoded if $r_e > I_{e}$. The probability of this event,  denoted as $\mc{O}_{e}$
\begin{equation} \label{eq:outage:e}
    P(\mc{O}_{e}) = \Pr\{ I_{e} < r_e \}.
\end{equation}
It is important to note that the assumption of complete knowledge of the CSI for the eMBB user constraints both  $P(\mc{O}_{u,e})$ and $P(\mc{O}_{e})$ to be either 1 or 0.

\paragraph*{URLLC outages}
According to our model, the data stream intended for $u$ is not successfully decoded if: a) the URLLC packet is wrongly decoded at receiver $u$, or b) the URLLC packet is not entirely received before the latency requirement $l^{\max}$.

a) A $u$ packet is wrongly decoded if $r_u > I_u$. The probability of this event denoted as $\mc{O}_u$ is 
\begin{equation} \label{eq:outage:u}
    P(\mc{O}_u)= \Pr \left\{ I_u \le r_u \right\} ,
\end{equation}
In our case it must be $P(\mc{O}_u)\le \epsilon_u$, i.e. $P(\mc{O}_u)$ needs to be smaller than the URLLC reliability constraint $\epsilon_u$.

b) The latency requirement is violated if the transmission is not successful within $l^{\max}$ mini-slots. Thus, the probability of this event, denoted as $\mc{L}$, can be computed  as~\cite{Durisi2016} 
\begin{equation} \label{eq:outage:L}
    \begin{aligned}
        P(\mc{L}) 
        &= \Pr\{I_u \le r_u \cup \Delta_u + M_u > l^{\max} \} \\
        &\le \Pr\{I_u \le r_u\} + \Pr\{M_u > l^{\max} - \Delta_u\} \\
        &\le \epsilon_u + \Pr\{M_u > l^{\max} - \Delta_u \},
    \end{aligned}
\end{equation}
where $\Delta_u$ is the number of mini-slots between the arrival and the first transmission of the URLLC packet. 
Hence, the latency constraint can be always satisfied choosing a reasonable value of $M_u$, such as
\begin{equation} 
\label{eq:M}
    M_u \le l^{\max} - \Delta_u,
\end{equation}
provided that  $P(\mc{O}_u)$  fulfills the reliability requirement.

\section{Power minimization problem}
In order to formulate the optimization process, we further denote with $P^\text{tot} = \sum_{t\in\mc{M}}\sum_{f\in\mc{F}} P_e(t, f)+ P_u(t,f)$ the overall power spent, with $\mb{P}_u$ the vector collecting all the URLLC power coefficients, and with $\mb{P}_e$ the vector collecting all the eMBB power coefficients. 
Putting together the SIC~\eqref{eq:outage:sic} and the rate~\eqref{eq:outage:e} requirements of the eMBB, and the reliability~\eqref{eq:outage:u} requirement of the URLLC user the resource allocation process can be formalized as follows
\begin{align}
\label{op:noma}
    \min_{\mb{P}_e, \mb{P}_u}& P^\text{tot} \\
    \text{s.t. } & P(\mc{O}_{e}) = 0, \tag{\ref{op:noma}.a} \label{op:noma:e}\\
    & P(\mc{O}_{u,e}) = 0, \tag{\ref{op:noma}.b} \label{op:noma:ue}\\
    & P(\mc{O}_u) \le \epsilon_u, \tag{\ref{op:noma}.c} \label{op:noma:u}\\
    & \mb{P}_u \succeq 0, \mb{P}_e \succeq 0. \tag{\ref{op:noma}.d}
\end{align}
For the OMA case, we remark that the problem is almost identical, except that eq.~\eqref{op:noma:ue} is substituted by a constraint on the orthogonality of resources.
Problem~\eqref{op:noma} is not convex with respect to both power coefficients. Moreover, the evaluation of~~\eqref{op:noma:u} is not known: only bound or approximations can be used, e.g. see~\cite{Li2020}. Unfortunately, the approximations make use of non-convex special functions, hard to be optimized.
In the following, we describe a simplified solution to the problem.

The general idea is to split the problem for $e$ and then for $u$, and solve simpler convex problems: we minimize the eMBB power subject to constraint~\eqref{op:noma:e}; we found the minimum URLLC power guaranteeing relation~\eqref{op:noma:ue}; we evaluate the URLLC power allocation constrained to~\eqref{op:noma:u}.
Having considered a single coherence time, the power analysis can be made assuming that $P_u(t,f)$ and $P_e(t,f)$ will not change on different mini-slots.
Moreover, given~\eqref{eq:M}, the time resources of URLLC transmissions are already defined. Thus, we can impose $P_i(t,f) = P_i(f)$, $\forall t \in \mc{M}_i$, $i \in\{e,u\}$.

\subsection{eMBB requirements}
For the eMBB traffic, we exploit the knowledge of the CSI to allocate the power satisfying a minimum rate requirement, i.e. guaranteeing a correct decoding process.
Following the definition~\eqref{eq:info:e}, we can obtain the minimum power spent that guarantees that $P(\mc{O}_{e}) = 0$, solving the following problem
\begin{equation}
\label{eq:op:e}
\begin{aligned}
    \min_{\mb{P}_e \succeq 0} &\sum_{f\in\mc{F}_e} P_e(f) \\
    \text{s.t. } &\frac{1}{F_e}\sum_{f\in\mc{F}_e} \log_2\left( 1 + \gamma_e(f) P_e(f)\right) \ge \bar{r}_e.
\end{aligned}
\end{equation}
where $\bar{r}_e = r_e / M_e / F_e$ is the average target rate obtained when the informative bits are spread onto the resources available for $e$.
The solution of problem~\eqref{eq:op:e} can be computed through the well-known water-filling approach, obtaining the desired power coefficients $\mb{P}_e$.

The decoding process relies on successful interference cancellation, which now depends only on the URLLC power coefficients. Hence, employing~\eqref{eq:info:ue} and~\eqref{eq:outage:sic}, we can obtain the minimum power to guarantee that $P(\mc{O}_{u,e}) = 0$ by solving the following problem
\begin{equation}
\label{eq:op:sic}
\begin{aligned}
    \min_{\mb{P}_u \succeq 0} &\sum_{f\in\mc{F}_{u}} P_u(f) \\
    \text{s.t. } &\frac{1}{F_u}\sum_{f\in\mc{F}_u} \log_2\left( 1 + \frac{\gamma_e(f) P_u(f)}{1 + \gamma_e(f) P_e(f)}\right) \ge \bar{r}_u.
\end{aligned}
\end{equation}
where $\bar{r}_u = r_u / M_u / F_u$ is the spectral efficiency obtained when the informative bits are spread onto the resources available for $u$. Also in this case, a water-filling approach is employed to obtain the solution, labeled as $\mb{P}_u^\text{SIC}$ to highlight that this is the minimum power needed for the SIC constraint. 
Therefore, the requirements of the eMBB traffic are satisfied if power coefficients $\mb{P}_e$, and $\mb{P}_u^\text{SIC}$ are used for the transmission of $e$ and $u$ data streams, respectively.

\subsection{URLLC requirements}
For the URLLC user, having solved the latency requirement with ~\eqref{eq:M}, we only need to focus on the reliability problem. Thus,~\eqref{eq:outage:u} can be rewritten as
\begin{equation}
    \label{eq:outage:ufull}
    \Pr \left\{ \frac{1}{F_u} \sum_{f\in\mc{F}_u} \log_2\left( 1 + \frac{\gamma_u(f)P_u(f)}{1 + \gamma_u(f) P_e(f)} \right) \le \bar{r}_u \right\} \le \epsilon_u.
\end{equation}
In the case of a single frequency resource, i.e. $F_u = 1$, the solution of~\eqref{eq:outage:ufull} is known~\cite{Wang2018}, but for $F_u > 1$, a closed form solution does not exist. 
An empirical approach to evaluate~\eqref{eq:outage:ufull} regardless of the number of frequency resources involved is outlined in the following Section. Moreover, two algorithms able to solve the allocation problem based on the evaluation of~\eqref{eq:outage:ufull} are proposed. 


\section{Simplified solutions for URLLC power allocation}
The outage probability~\eqref{eq:outage:ufull} depends on $r_u$, $\rho_u$, $F_u$, $\mb{P}_e$ and $\mb{P}_u$. In theory, we can tabulate the outage probability as a function of these parameters. However, the variation of each element of $\mb{P}_e$ and $\mb{P}_u$ cannot be done in practice due to the huge number of trials to be done. Thus, we tabulate the outage, assuming the same $P_u$ and $P_e$ are used for all the considered FRs.
Let us denote as $\hat{O}_u(P_u, P_e, \rho_u, F_u, r_u )$ the Monte Carlo estimation of the outage probability $P(\mc{O}_u)$ when $P_u(f) = P_u$ and $P_e(f) = P_e$ for every $f\in\mc{F}_u$. 
In the following proposed algorithms, we use this table to solve the URLLC power allocation problem.

\subsection{Feasible algorithm}
Let us now consider to have computed the eMBB power coefficients $P_e(f)$ through solution of~\eqref{eq:op:e}. 
Given $\rho_u$, $r_u$ and $F_u$, we select the resource of the worst interference seen, i.e. $f^* = \argmax_f P_e(f)$. Then, we extract from table $\hat{\mc{O}}_u$ the minimum $P_u$ that guarantees $\hat{O}_u(P_u, P_e(f^*), \rho_u, F_u, r_u ) \le \epsilon_u$. In other words, we extract the URLLC power assuming that all channels experience the strongest interference. Finally, we set:
\begin{equation}
    \label{eq:alg:worst}
    P_u^*(f) = \max\left\{P_u, P_u^\text{SIC}(f)\right\}, \quad \forall f\in\mc{F}_u.
\end{equation}
In the following, we prove that the algorithm summarized in Table~\ref{alg:N-fea} provides a feasible solution to the power allocation problem~\eqref{op:noma}.
\begin{proposition}
Algorithm~\ref{alg:N-fea} guarantees both URLLC reliability requirement~\eqref{eq:outage:u} and eMBB SIC requirement~\eqref{eq:outage:sic}.
\end{proposition}
\begin{proof}
Power coefficient $P_u$ is the minimum power ensuring $\hat{\mc{O}}_u(P_u, P_e(f^*), \rho_u, F_u, r_u ) \le \epsilon_u$. In other words, a transmission employing power coefficients $P_u$ and interference coefficients $P_e(f^*) \ge P_e(f)$, $\forall f\in\mc{F}_u$ has outage probability $P(\mc{O}_u) \le \epsilon_u$. Thus, the actual mutual information evaluated with the same power coefficients $\mb{P}_u$ but smaller (or equal) interference coefficients $\mb{P}_e$ can only be greater (or equal). Therefore, the actual outage probability~\eqref{eq:outage:u} can be only lower (or equal), satisfying the reliability constraint. Furthermore, operator~\eqref{eq:alg:worst} assures that the power coefficient is increased to meet the minimum power needed for the SIC, if needed. Again, this operation can only increase the mutual information~\eqref{eq:info:u}, satisfying at the same time the SIC requirement.
\end{proof}

\begin{algorithm}
\footnotesize
\caption{Feasible algorithm (N-fea)}
\label{alg:N-fea}
\textbf{Initialize:}
Having populated the table $\hat{O}_u(P_u, P_e, \rho_u, F_u, r_u ) $\; 
Compute $\mb{P}_e$ solving~\eqref{eq:op:e} through water-filling approach\;
Compute $\mb{P}_u^\text{SIC}$ solving~\eqref{eq:op:sic} through water-filling approach\;
$f^* = \argmax_f P_e(f)$\;
$P_u = \min \{P_u \,|\, \hat{O}_u(P_u, P_e(f^*), \rho_u, F_u, r_u ) \le \epsilon_u\}$\;
\For {$f\in\mc{F}_u$}{
$P_u^*(f) = \max\left\{P_u, P_u^\text{SIC}(f)\right\}$\;
}
\textbf{Output:} $\mb{P}_u^*$
\end{algorithm}

\subsection{Heuristic algorithm}
The feasible algorithm is sub-optimal in terms of power spent, considering that we constraint all channels to act as the worst one. 
To overcome this limitation, we propose a different heuristic approach. Given $\rho_u$, $r_u$ and $F_u$, we perform the allocation per FR: we extract from table $\hat{\mc{O}}_u$ the minimum $P_u(f)$ that guarantees $\hat{O}_u(P_u(f), P_e(f), \rho_u, F_u, r_u ) \le \epsilon_u$. In other words, we deploy the power for channel $f\in\mc{F}_u$ assuming that all other channels have inference power $P_e(f)$.
Finally, for each FR, we set:
\begin{equation}
    \label{eq:alg:all}
    P_u^*(f) = \max\left\{P_u(f), P_u^\text{SIC}(f)\right\}, \quad \forall f\in\mc{F}_u.
\end{equation}
Even if this scheme is still sub-optimal, it saves power with respect to the previous algorithm. Note that there is no theoretical guarantee that this heuristic will satisfy the reliability constraint. However, experimental results will show that using this scheme the reliability constraint is always verified for the parameters of interest.

\begin{algorithm}
\footnotesize
\caption{Heuristic algorithm (N-heu)}
\label{alg:N-heu}
\textbf{Initialize:}
Having populated the table $\hat{O}_u(P_u, P_e, \rho_u, F_u, r_u ) $\; 
Compute $\mb{P}_e$ solving~\eqref{eq:op:e} through water-filling approach\;
Compute $\mb{P}_u^\text{SIC}$ solving~\eqref{eq:op:sic} through water-filling approach\;
\For {$f\in\mc{F}_u$}{
$P_u(f) = \min \{P_u(f) \,|\, \hat{O}_u(P_u(f), P_e(f), \rho_u, F_u, r_u ) \le \epsilon_u\}$\;
$P_u^*(f) = \max\left\{P_u(f), P_u^\text{SIC}(f)\right\}$\;
}
\textbf{Output:}  $\mb{P}_u^*$
\end{algorithm}
\normalsize

\section{Numerical simulations}
\label{sec:results}
In this section, a comparison in terms of power consumption of using NOMA and OMA for eMBB and URLLC traffics is presented.

Increasing the number of mini-slots linearly reduce the average target rate on each resources, while increasing the number of frequency resources gives a not negligible diversity gain. Thus, without loss of generality, we consider a resource grid formed by a single mini-slot $M = 1$ and $F = 12$ FRs. Furthermore, we set $r_u = 1$ and $r_e = 6$ bit/s/Hz; the noise power of both user is $\sigma_i^2 = - 92$ dBm; the reliability requirement is set as $\epsilon_u = 10^{-5}$.
Each instance of simulation is made placing the users in a cell of 100 m of radius and computing the power consumption. 
The positioning of the two users is made as follows: the URLLC user $u$ is placed at a fixed distance $d_u$ [m] from the BS; the eMBB user is then randomly positioned following a uniform distribution in the cell. Then, we compute the path loss assuming no shadowing and setting the sum of antenna gains to 17.15 dB and the path loss exponent to 4. Finally, the results are averaged for all the different positions of the eMBB user.

With respect to the OMA paradigm, we apply the same procedure described for NOMA taking into account the few differences between the two schemes. Firstly, we ignore the SIC requirement, which is a trait of the NOMA only, i.e. we set $P_u^\text{SIC}(f) = 0$, $\forall f \in \mc{F}_u$. 
Then, we remark that, for OMA allocation, Algorithms~\ref{alg:N-fea} and~\ref{alg:N-heu} give the same results. Indeed, the interference power coefficient of $e$ user is zero for every resources given to $u$, i.e. $P_e(f) = 0$, $\forall f\in\mc{F}_u$.
Before comparing with NOMA, we studied the performance of OMA reserving various percentages of resources to the URLLC. For this setting, we found that the best performances of OMA are attained in two cases: a) reserving 50\% of the resources to URLLC for farther users, and b) reserving 25\% of the resources for URLLC, for closer users. We label the result of a) as OMA-50 and the results of b) as OMA-25. 
In the case of NOMA, we show the results for both  Algorithms~\ref{alg:N-fea} and~\ref{alg:N-heu}, labeled as N-fea and N-heu, respectively. 

Fig.~\ref{fig:tot_power_re6} shows the total power spent $P^\text{tot}$ as a function of $d_u$, while the corresponding estimated $P(\mc{O}_u)$ is presented in Fig.~\ref{fig:outage_re6}. $P(\mc{O}_u)$ is evaluated by Monte Carlo simulation keeping fixed the allocated power coefficients of each instance while the channel gains of $u$ are randomly generated with the corresponding path loss given by the distance $d_u$. Moreover, the tabulated outage N-tab for the NOMA paradigm used for the evaluation of power is also plotted in Fig.~\ref{fig:outage_re6}, to show the difference between the tabulated and estimated outage in the NOMA case. The tabulated outage for the OMA case is not presented because identical to the estimated one. 
We can see that NOMA assures a lower power consumption with respect to the OMA paradigms, guaranteeing at the same time the reliability requirement. As already stated, the percentage of resources to use in case OMA depends on $d_u$. In particular, OMA-25 obtain better performance if $d_u < 29$ m, while better performance is given by OMA-50 for the other distances presented. 

\begin{figure}[tb]
    \centering
    \begin{subfigure}[c]{\columnwidth}
        \centering
\begin{tikzpicture}

\definecolor{color0}{rgb}{1,0.270588235294118,0}
\definecolor{color1}{rgb}{1,0.549019607843137,0}
\definecolor{color2}{rgb}{1,0.843137254901961,0}

\begin{axis}[
legend cell align={left},
legend style={
  fill opacity=0.8,
  draw opacity=1,
  text opacity=1,
  at={(0.03,0.97)},
  anchor=north west,
  draw=white!80!black
},
tick align=outside,
tick pos=left,
x grid style={white!69.0196078431373!black},
xlabel={\(\displaystyle d_u\) [m]},
xmajorgrids,
xmin=15.7374251466673, xmax=76.3826446482485,
xtick style={color=black},
xtick={10,20,30,40,50,60,70,80},
xticklabels={
  \(\displaystyle {10}\),
  \(\displaystyle {20}\),
  \(\displaystyle {30}\),
  \(\displaystyle {40}\),
  \(\displaystyle {50}\),
  \(\displaystyle {60}\),
  \(\displaystyle {70}\),
  \(\displaystyle {80}\)
},
y grid style={white!69.0196078431373!black},
ylabel={\(\displaystyle P^{tot}\) [dBm]},
ymajorgrids,
ymin=7.9, ymax=16.3,
xmin=25, xmax=71,
ytick style={color=black},
ytick={6,8,10,12,14,16,18},
yticklabels={
  \(\displaystyle {6}\),
  \(\displaystyle {8}\),
  \(\displaystyle {10}\),
  \(\displaystyle {12}\),
  \(\displaystyle {14}\),
  \(\displaystyle {16}\),
  \(\displaystyle {18}\)
}
]
\addplot [semithick, blue, mark=x]
table {%
73.626043761813 12.0892869866435
69.5074673912429 11.5491906562891
65.619280573792 11.0651624989143
61.9485955197459 10.638984385619
58.4832453710536 10.2555342518559
55.2117438730415 9.92881212885444
52.1232473020575 9.64081916745454
49.2075185228483 9.40648975544556
46.454893056536 9.1917585826807
43.8562470467222 9.01605463478315
41.4029670175405 8.87419951892167
39.0869213234165 8.75119543624051
36.9004331959035 8.64981837673711
34.8362552982549 8.58497940474448
32.8875457033906 8.52330092606585
31.0478452156363 8.47122449653346
29.3110559610633 8.4333389969812
27.6714211754672 8.40577023739167
26.1235061229884 8.38337469235995
24.6621800821291 8.35656186820285
23.2825993394561 8.34039435343495
21.9801911346209 8.32958247595836
20.7506385034823 8.31657259270518
19.58986596909 8.29589357278671
18.4940260331029 8.28059314602966
};
\addlegendentry{N-fea}
\addplot [semithick, green!50!black, mark=x, dashed, mark options={solid}]
table {%
73.626043761813 11.7271832696498
69.5074673912429 11.1426461688773
65.619280573792 10.604072875891
61.9485955197459 10.1311411889958
58.4832453710536 9.69154852991264
55.2117438730415 9.34509178415764
52.1232473020575 9.00363639619399
49.2075185228483 8.72524868106957
46.454893056536 8.50505848345739
43.8562470467222 8.30052769960881
41.4029670175405 8.16390133235985
39.0869213234165 8.083640789819
36.9004331959035 8.083640789819
34.8362552982549 8.083640789819
32.8875457033906 8.083640789819
31.0478452156363 8.083640789819
29.3110559610633 8.083640789819
27.6714211754672 8.083640789819
26.1235061229884 8.083640789819
24.6621800821291 8.083640789819
23.2825993394561 8.083640789819
21.9801911346209 8.083640789819
20.7506385034823 8.083640789819
19.58986596909 8.083640789819
18.4940260331029 8.083640789819
};
\addlegendentry{N-heu}
\addplot [semithick, color0, mark=star]
table {%
73.626043761813 16.8617940646314
69.5074673912429 16.8617940646314
65.619280573792 16.8617940646314
61.9485955197459 16.8617940646314
58.4832453710536 16.8617940646314
55.2117438730415 16.8617940646314
52.1232473020575 16.6796847423516
49.2075185228483 15.7475446893988
46.454893056536 14.988488122807
43.8562470467222 14.189099833085
41.4029670175405 13.4064347118361
39.0869213234165 12.7642362332738
36.9004331959035 12.1134454052999
34.8362552982549 11.4806637636173
32.8875457033906 10.9667154450485
31.0478452156363 10.5165409981372
29.3110559610633 10.0508178491626
27.6714211754672 9.69031183668244
26.1235061229884 9.4041127103242
24.6621800821291 9.14283122764043
23.2825993394561 8.91240749752166
21.9801911346209 8.7188403156318
20.7506385034823 8.56778423606622
19.58986596909 8.43776232813922
18.4940260331029 8.33086727032591
};
\addlegendentry{OMA-25}
\addplot [semithick, color1, mark=star, dashed]
table {%
73.626043761813 15.3065107486805
69.5074673912429 14.5493137451919
65.619280573792 13.924860657418
61.9485955197459 13.2999474954759
58.4832453710536 12.7494832463765
55.2117438730415 12.2526170261953
52.1232473020575 11.8145068440499
49.2075185228483 11.4344258907247
46.454893056536 11.1044335360042
43.8562470467222 10.8088909618352
41.4029670175405 10.589352975221
39.0869213234165 10.3919197265518
36.9004331959035 10.2202461283559
34.8362552982549 10.0778474589315
32.8875457033906 9.9678981639467
31.0478452156363 9.89302111386523
29.3110559610633 9.81683042099939
27.6714211754672 9.73927916791339
26.1235061229884 9.69997797484906
24.6621800821291 9.69997797484906
23.2825993394561 9.69997797484906
21.9801911346209 9.69997797484906
20.7506385034823 9.69997797484906
19.58986596909 9.69997797484906
18.4940260331029 9.69997797484906
};
\addlegendentry{OMA-50}
\end{axis}

\end{tikzpicture}
        \caption{Total power spent.}
        \label{fig:tot_power_re6}
    \end{subfigure}
    ~
    \begin{subfigure}[c]{\columnwidth}
        \centering
\begin{tikzpicture}

\definecolor{color0}{rgb}{0,0,0}
\definecolor{color1}{rgb}{1,0.270588235294118,0}
\definecolor{color2}{rgb}{1,0.549019607843137,0}
\definecolor{color3}{rgb}{1,0.843137254901961,0}

\begin{axis}[
legend cell align={left},
legend style={
  fill opacity=0.8,
  draw opacity=1,
  text opacity=1,
  at={(1,0)},
  anchor=south east,
  draw=white!80!black
},
log basis y={10},
tick align=outside,
tick pos=left,
x grid style={white!69.0196078431373!black},
xlabel={\(\displaystyle d_u\) [m]},
xmajorgrids,
xmin=25, xmax=71,
xtick style={color=black},
y grid style={white!69.0196078431373!black},
ylabel={\(\displaystyle P(\mathcal{O}_u)\)},
ymajorgrids,
ymin=1e-08, ymax=1e-4,
ymode=log,
ytick style={color=black}
]
\addplot [semithick, blue, mark=x]
table {%
73.626043761813 3.865e-06
69.5074673912429 3.448e-06
65.619280573792 2.974e-06
61.9485955197459 2.505e-06
58.4832453710536 2.222e-06
55.2117438730415 1.774e-06
52.1232473020575 1.656e-06
49.2075185228483 1.244e-06
46.454893056536 1.115e-06
43.8562470467222 8.19e-07
41.4029670175405 6.19e-07
39.0869213234165 3.93e-07
36.9004331959035 7.7e-08
34.8362552982549 1.3e-08
32.8875457033906 1e-09
31.0478452156363 0
29.3110559610633 0
27.6714211754672 0
26.1235061229884 0
24.6621800821291 0
23.2825993394561 0
21.9801911346209 0
20.7506385034823 0
19.58986596909 0
18.4940260331029 0
};
\addlegendentry{N-fea}
\addplot [semithick, green!50!black, mark=x, dashed, mark options={solid}]
table {%
73.626043761813 7.749e-06
69.5074673912429 7.426e-06
65.619280573792 7.364e-06
61.9485955197459 7.086e-06
58.4832453710536 7.683e-06
55.2117438730415 6.46e-06
52.1232473020575 7.646e-06
49.2075185228483 7.914e-06
46.454893056536 7.112e-06
43.8562470467222 8.062e-06
41.4029670175405 5.533e-06
39.0869213234165 4.393e-06
36.9004331959035 1.213e-06
34.8362552982549 3.03e-07
32.8875457033906 8.8e-08
31.0478452156363 2.9e-08
29.3110559610633 1e-08
27.6714211754672 2e-09
26.1235061229884 0
24.6621800821291 0
23.2825993394561 0
21.9801911346209 0
20.7506385034823 0
19.58986596909 0
18.4940260331029 0
};
\addlegendentry{N-heu}
\addplot [semithick, color0]
table {%
73.626043761813 8.58901338833537e-06
69.5074673912429 8.12905527345437e-06
65.619280573792 7.9340026203403e-06
61.9485955197459 7.83047637203218e-06
58.4832453710536 7.82441578840617e-06
55.2117438730415 7.46760794003592e-06
52.1232473020575 7.32100678497161e-06
49.2075185228483 7.21620252197623e-06
46.454893056536 6.22703596960208e-06
43.8562470467222 7.19104577580345e-06
41.4029670175405 4.92900598577172e-06
39.0869213234165 3.31855544313899e-06
36.9004331959035 1.71547971263618e-06
34.8362552982549 1.1439704466716e-06
32.8875457033906 1.03702731070895e-06
31.0478452156363 8.26091881745322e-07
29.3110559610633 6.86581353447364e-07
27.6714211754672 6.10908092974009e-07
26.1235061229884 5.82379432735472e-07
24.6621800821291 5.4446753965916e-07
23.2825993394561 4.877904569877e-07
21.9801911346209 3.61285553820648e-07
20.7506385034823 3.54080071910692e-07
19.58986596909 3.38273611368627e-07
18.4940260331029 3.356527464228e-07
};
\addlegendentry{N-tab}
\addplot [semithick, color1, mark=star]
table {%
73.626043761813 1e-05
69.5074673912429 1e-05
65.619280573792 1e-05
61.9485955197459 1e-05
58.4832453710536 1e-05
55.2117438730415 1e-05
52.1232473020575 1e-05
49.2075185228483 1.0e-05
46.454893056536 9.821e-06
43.8562470467222 9.719e-06
41.4029670175405 1.0159e-05
39.0869213234165 9.487e-06
36.9004331959035 9.629e-06
34.8362552982549 9.944e-06
32.8875457033906 9.888e-06
31.0478452156363 9.459e-06
29.3110559610633 1.0072e-05
27.6714211754672 1.006e-05
26.1235061229884 9.49e-06
24.6621800821291 9.123e-06
23.2825993394561 9.292e-06
21.9801911346209 9.476e-06
20.7506385034823 9.323e-06
19.58986596909 9.38e-06
18.4940260331029 9.735e-06
};
\addlegendentry{OMA-25}
\addplot [semithick, color2, mark=star, dashed]
table {%
73.626043761813 9.07e-06
69.5074673912429 9.645e-06
65.619280573792 8.993e-06
61.9485955197459 9.294e-06
58.4832453710536 9.268e-06
55.2117438730415 9.119e-06
52.1232473020575 9.107e-06
49.2075185228483 8.93e-06
46.454893056536 8.822e-06
43.8562470467222 9.186e-06
41.4029670175405 8.056e-06
39.0869213234165 8.055e-06
36.9004331959035 8.044e-06
34.8362552982549 8.406e-06
32.8875457033906 8.047e-06
31.0478452156363 6.199e-06
29.3110559610633 5.827e-06
27.6714211754672 6.233e-06
26.1235061229884 6.24e-06
24.6621800821291 1.727e-06
23.2825993394561 4.41e-07
21.9801911346209 1.15e-07
20.7506385034823 3.9e-08
19.58986596909 8e-09
18.4940260331029 1e-09
};
\addlegendentry{OMA-50}
\end{axis}

\end{tikzpicture}
        \caption{Tabulated and estimated outage.}
        \label{fig:outage_re6}
    \end{subfigure}
    \caption{Average results obtained as a function of URLLC distance $d_u$.}
    \label{fig:power-outage6}
\end{figure}
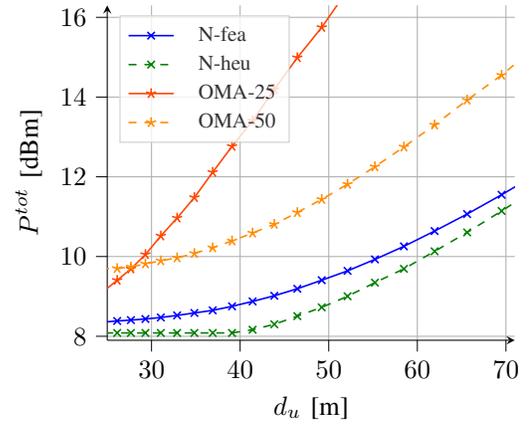
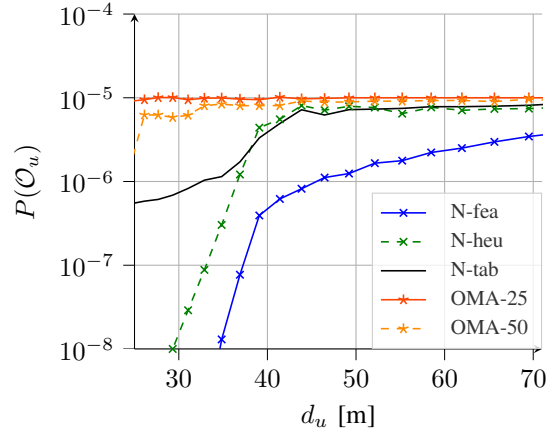
\begin{table}[bth]
    \centering
    \caption{Average eMBB power spent in dBm.}
    \begin{tabular}{rrr} 
\toprule
  NOMA &   OMA-25 &   OMA-50 \\ 
\midrule
   7.02 &  7.90 &  9.50 \\ 
\bottomrule
\end{tabular}

    \label{tab:embb_power}
\end{table}

To fully understand these results, we present in Table~\ref{tab:embb_power} the average value of the total power spent for the eMBB user, which does not depends on $d_u$. The pre-allocation of resources leads to an increase in power: more resources are reserved for URLLC only, more power must be allocated to the eMBB to compensate for the frequency diversity gain loss.
In Fig.~\ref{fig:urllc_powers} we show the power spent for the URLLC user, i.e. $P_u^\text{tot} = \sum_{f\in\mc{F}_u} P_u(f)$. In particular, we show: the results of N-fea and N-heu before the application $\max$ operator in~\eqref{eq:alg:worst} and~\eqref{eq:alg:all}; the average minimum power needed to satisfy the SIC requirement, labeled as N-SIC; the power spent for the OMA allocations. Unsurprisingly, the power spent decreases with the increase of resources reserved for the URLLC. We can clearly see that the power consumption of N-heu is lower than the OMA when $d_u \ge 45$ m. However, for $d_u \le 40$ m the power needed by the SIC limits the power consumption for N-heu.

\begin{figure}[htb]
    \centering
\begin{tikzpicture}

\definecolor{color0}{rgb}{1,0.270588235294118,0}
\definecolor{color1}{rgb}{1,0.549019607843137,0}
\definecolor{color2}{rgb}{1,0.843137254901961,0}

\begin{axis}[
legend cell align={left},
legend style={
  fill opacity=0.8,
  draw opacity=1,
  text opacity=1,
  at={(1,0)},
  anchor=south east,
  draw=white!80!black
},
tick align=outside,
tick pos=left,
x grid style={white!69.0196078431373!black},
xlabel={\(\displaystyle d_u\) [m]},
xmajorgrids,
xmin=25, xmax=71,
xtick style={color=black},
xtick={10,20,30,40,50,60,70,80},
xticklabels={
  \(\displaystyle {10}\),
  \(\displaystyle {20}\),
  \(\displaystyle {30}\),
  \(\displaystyle {40}\),
  \(\displaystyle {50}\),
  \(\displaystyle {60}\),
  \(\displaystyle {70}\),
  \(\displaystyle {80}\)
},
y grid style={white!69.0196078431373!black},
ylabel={\(\displaystyle P_{u}^\text{tot}\) [dBm]},
ymajorgrids,
ymin=-6.4816448671918, ymax=16,
ytick style={color=black},
ytick={-10,-5,0,5,10,15,20},
yticklabels={
  \(\displaystyle {-10}\),
  \(\displaystyle {-5}\),
  \(\displaystyle {0}\),
  \(\displaystyle {5}\),
  \(\displaystyle {10}\),
  \(\displaystyle {15}\),
  \(\displaystyle {20}\)
}
]
\addplot [semithick, blue, mark=x]
table {%
73.626043761813 10.47048901933
69.5074673912429 9.66251255508438
65.619280573792 8.89056539691924
61.9485955197459 8.16315559445646
58.4832453710536 7.46018906269003
55.2117438730415 6.81555627186673
52.1232473020575 6.20374286560307
49.2075185228483 5.66875297571841
46.454893056536 5.14272040610058
43.8562470467222 4.68163645502884
41.4029670175405 4.28547627753773
39.0869213234165 3.92180164344591
36.9004331959035 3.60592449427184
34.8362552982549 3.39531620845357
32.8875457033906 3.1881412462672
31.0478452156363 3.00761048272649
29.3110559610633 2.87281861731403
27.6714211754672 2.77279952423132
26.1235061229884 2.69029876997702
24.6621800821291 2.58999535506777
23.2825993394561 2.52868217455281
21.9801911346209 2.48732048268639
20.7506385034823 2.43716079341186
19.58986596909 2.35653712818884
18.4940260331029 2.2961575782157
};
\addlegendentry{N-fea}
\addplot [semithick, green!50!black, mark=x, dashed, mark options={solid}]
table {%
73.626043761813 9.93436684959524
69.5074673912429 9.01767514210051
65.619280573792 8.10122576057528
61.9485955197459 7.22026895135252
58.4832453710536 6.31489972667329
55.2117438730415 5.52214033127735
52.1232473020575 4.64786353549604
49.2075185228483 3.84242577498408
46.454893056536 3.12550652836172
43.8562470467222 2.37470176579609
41.4029670175405 1.81357911401182
39.0869213234165 1.24900147511208
36.9004331959035 1.0162626970867
34.8362552982549 0.861432596166284
32.8875457033906 0.722769931798308
31.0478452156363 0.62550240852827
29.3110559610633 0.545005636884805
27.6714211754672 0.479516411745037
26.1235061229884 0.426961186895139
24.6621800821291 0.37714146304005
23.2825993394561 0.337551901569562
21.9801911346209 0.317620958745318
20.7506385034823 0.284964930639937
19.58986596909 0.25727337144992
18.4940260331029 0.229986548646358
};
\addlegendentry{N-heu}
\addplot [semithick, white!50!black, dashed]
table {%
73.626043761813 1.45607584651212
69.5074673912429 1.45607584651212
65.619280573792 1.45607584651212
61.9485955197459 1.45607584651212
58.4832453710536 1.45607584651212
55.2117438730415 1.45607584651212
52.1232473020575 1.45607584651212
49.2075185228483 1.45607584651212
46.454893056536 1.45607584651212
43.8562470467222 1.45607584651212
41.4029670175405 1.45607584651212
39.0869213234165 1.45607584651212
36.9004331959035 1.45607584651212
34.8362552982549 1.45607584651212
32.8875457033906 1.45607584651212
31.0478452156363 1.45607584651212
29.3110559610633 1.45607584651212
27.6714211754672 1.45607584651212
26.1235061229884 1.45607584651212
24.6621800821291 1.45607584651212
23.2825993394561 1.45607584651212
21.9801911346209 1.45607584651212
20.7506385034823 1.45607584651212
19.58986596909 1.45607584651212
18.4940260331029 1.45607584651212
};
\addlegendentry{N-SIC}
\addplot [semithick, color0, mark=star]
table {%
73.626043761813 16.2712125471966
69.5074673912429 16.2712125471966
65.619280573792 16.2712125471966
61.9485955197459 16.2712125471966
58.4832453710536 16.2712125471966
55.2117438730415 16.2712125471966
52.1232473020575 16.0619271487815
49.2075185228483 14.9678748622465
46.454893056536 14.0425695048209
43.8562470467222 13.0241690723149
41.4029670175405 11.9692967610872
39.0869213234165 11.0467702121334
36.9004331959035 10.0413272355019
34.8362552982549 8.97068710711596
32.8875457033906 8.0074203490595
31.0478452156363 7.06737012541293
29.3110559610633 5.96038715047598
27.6714211754672 4.96787486224629
26.1235061229884 4.05707016877311
24.6621800821291 3.0879700386923
23.2825993394561 2.07339363110432
21.9801911346209 1.04677021213279
20.7506385034823 0.0776700820524826
19.58986596909 -0.936906325536297
18.4940260331029 -1.96352974450702
};
\addlegendentry{OMA-25}
\addplot [semithick, color1, mark=star, dashed]
table {%
73.626043761813 13.9840777813582
69.5074673912429 12.9219774204979
65.619280573792 11.9809870637558
61.9485955197459 10.9590309690583
58.4832453710536 9.96771623903802
55.2117438730415 8.97068710711579
52.1232473020575 7.9781748188861
49.2075185228483 6.99437773799787
46.454893056536 6.00683620093333
43.8562470467222 4.96787486224587
41.4029670175405 4.0570701687726
39.0869213234165 3.0879700386923
36.9004331959035 2.07339363110351
34.8362552982549 1.04677021213279
32.8875457033906 0.0776700820512019
31.0478452156363 -0.714142378426069
29.3110559610633 -1.68324250850433
27.6714211754672 -2.93262987458695
26.1235061229884 -3.7244423350632
};
\addlegendentry{OMA-50}
\end{axis}

\end{tikzpicture}
    \caption{Average URLLC power spent as a function of the distance $d_u$.}
    \label{fig:urllc_powers}
\end{figure}
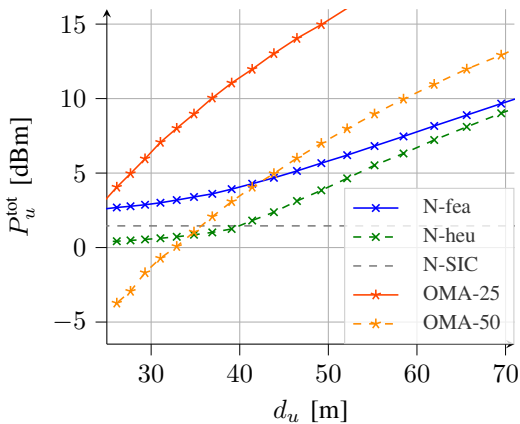

\section{Conclusions}
\label{sec:conclusions}
In this paper, the comparison in terms of power consumption for OMA and NOMA paradigms applied on eMBB and URLLC users is studied. Our model assumes a perfect CSI knowledge for the eMBB user and a statistical CSI knowledge of the URLLC user. Under these assumptions, the general power allocation process of both traffics can be split into several sub-problems. Using the knowledge of the eMBB CSI, we provide a practical solution for the minimum power coefficients.
Numerical results prove that the NOMA approach assures a lower power consumption with respect to the OMA paradigm, guaranteeing at the same time the minimum requirements of both users.

\bibliographystyle{IEEEtranNoUrl}
\bibliography{IEEEabrv, slicing.bib}

\begin{thebibliography}{10}
\providecommand{\url}[1]{#1}
\csname url@samestyle\endcsname
\providecommand{\newblock}{\relax}
\providecommand{\bibinfo}[2]{#2}
\providecommand{\BIBentrySTDinterwordspacing}{\spaceskip=0pt\relax}
\providecommand{\BIBentryALTinterwordstretchfactor}{4}
\providecommand{\BIBentryALTinterwordspacing}{\spaceskip=\fontdimen2\font plus
\BIBentryALTinterwordstretchfactor\fontdimen3\font minus
  \fontdimen4\font\relax}
\providecommand{\BIBforeignlanguage}[2]{{%
\expandafter\ifx\csname l@#1\endcsname\relax
\typeout{** WARNING: IEEEtran.bst: No hyphenation pattern has been}%
\typeout{** loaded for the language `#1'. Using the pattern for}%
\typeout{** the default language instead.}%
\else
\language=\csname l@#1\endcsname
\fi
#2}}
\providecommand{\BIBdecl}{\relax}
\BIBdecl

\bibitem{Zhang2017}
H.~{Zhang}, N.~{Liu}, X.~{Chu}, K.~{Long}, A.~H. {Aghvami}, and V.~C.~M.
  {Leung}, ``Network slicing based {5G} and future mobile networks: Mobility,
  resource management, and challenges,'' \emph{IEEE Commun. Mag.}, vol.~55,
  no.~8, pp. 138--145, 2017.

\bibitem{Doro2019}
S.~{D’Oro}, F.~{Restuccia}, A.~{Talamonti}, and T.~{Melodia}, ``The slice is
  served: Enforcing radio access network slicing in virtualized 5g systems,''
  in \emph{IEEE INFOCOM 2019 - IEEE Conference on Computer Communications},
  2019, pp. 442--450.

\bibitem{Anand2020}
A.~{Anand}, G.~{de Veciana}, and S.~{Shakkottai}, ``Joint scheduling of urllc
  and embb traffic in 5g wireless networks,'' \emph{IEEE/ACM Trans. Netw.},
  vol.~28, no.~2, pp. 477--490, 2020.

\bibitem{Elsayed2019}
M.~{Elsayed} and M.~{Erol-Kantarci}, ``{AI-Enabled Radio Resource Allocation in
  5G for URLLC and eMBB Users},'' in \emph{2019 IEEE 2nd 5G World Forum
  (5GWF)}, 2019, pp. 590--595.

\bibitem{Alsenwi2021}
M.~{Alsenwi}, N.~H. {Tran}, M.~{Bennis}, S.~R. {Pandey}, A.~K. {Bairagi}, and
  C.~S. {Hong}, ``Intelligent resource slicing for {eMBB} and {URLLC}
  coexistence in {5G} and beyond: A deep reinforcement learning based
  approach,'' \emph{IEEE Trans. Wireless Commun.}, pp. 1--1, 2021.

\bibitem{3gpp:access}
\BIBentryALTinterwordspacing
3GPP, ``{Study on New Radio (NR) access technology},'' {3rd Generation
  Partnership Project (3GPP)}, Technical Report (TR) 38.912, 05 2017, version
  14.0.0.
\BIBentrySTDinterwordspacing

\bibitem{Xu2015}
\BIBentryALTinterwordspacing
P.~Xu, Z.~Ding, X.~Dai, and H.~V. Poor, ``{NOMA:} an information theoretic
  perspective,'' \emph{CoRR}, vol. abs/1504.07751, 2015.
\BIBentrySTDinterwordspacing

\bibitem{Popovski2018}
P.~{Popovski}, K.~F. {Trillingsgaard}, O.~{Simeone}, and G.~{Durisi}, ``5{G}
  wireless network slicing for {eMBB}, {URLLC}, and {mMTC}: A
  communication-theoretic view,'' \emph{IEEE Access}, vol.~6, pp.
  55\,765--55\,779, 2018.

\bibitem{Kalor2019}
A.~E. {Kalør} and P.~{Popovski}, ``Ultra-reliable communication for services
  with heterogeneous latency requirements,'' in \emph{2019 IEEE Globecom
  Workshops (GC Wkshps)}, 2019, pp. 1--6.

\bibitem{Li2019}
Y.~{Li}, C.~{Hu}, J.~{Wang}, and M.~{Xu}, ``Optimization of {URLLC} and {eMBB}
  multiplexing via deep reinforcement learning,'' in \emph{2019 IEEE/CIC
  International Conference on Communications Workshops in China (ICCC
  Workshops)}, 2019, pp. 245--250.

\bibitem{Bai2013}
B.~{Bai}, W.~{Chen}, K.~B. {Letaief}, and Z.~{Cao}, ``Outage exponent: A
  unified performance metric for parallel fading channels,'' \emph{IEEE Trans.
  Inf. Theory}, vol.~59, no.~3, pp. 1657--1677, 2013.

\bibitem{Coon2015}
J.~P. Coon, D.~E. Simmons, and M.~D. Renzo, ``Approximating the outage
  probability of parallel fading channels,'' \emph{IEEE Commun. Lett.},
  vol.~19, no.~12, pp. 2190--2193, 2015.

\bibitem{Li2020}
S.~{Li}, M.~{Derakhshani}, S.~{Lambotharan}, and L.~{Hanzo}, ``Outage
  probability analysis for the multi-carrier {NOMA} downlink relying on
  statistical {CSI},'' \emph{IEEE Trans. Commun.}, vol.~68, no.~6, pp.
  3572--3587, 2020.

\bibitem{She2017}
C.~{She}, C.~{Yang}, and T.~Q.~S. {Quek}, ``Radio resource management for
  {Ultra-Reliable} and {Low-Latency Communications},'' \emph{IEEE Commun.
  Mag.}, vol.~55, no.~6, pp. 72--78, 2017.

\bibitem{Polyanskiy2010}
Y.~{Polyanskiy}, H.~V. {Poor}, and S.~{Verdu}, ``Channel coding rate in the
  finite blocklength regime,'' \emph{IEEE Trans. Inf. Theory}, vol.~56, no.~5,
  pp. 2307--2359, 2010.

\bibitem{Yang2014}
W.~{Yang}, G.~{Durisi}, T.~{Koch}, and Y.~{Polyanskiy}, ``Quasi-static
  multiple-antenna fading channels at finite blocklength,'' \emph{IEEE Trans.
  Inf. Theory}, vol.~60, no.~7, pp. 4232--4265, 2014.

\bibitem{Durisi2016}
G.~{Durisi}, T.~{Koch}, and P.~{Popovski}, ``Toward massive, ultra reliable,
  and low-latency wireless communication with short packets,'' \emph{Proc.
  IEEE}, vol. 104, no.~9, pp. 1711--1726, 2016.

\bibitem{Wang2018}
X.~{Wang}, J.~{Wang}, L.~{He}, and J.~{Song}, ``Outage analysis for downlink
  {NOMA} with statistical channel state information,'' \emph{IEEE Wireless
  Commun. Lett.}, vol.~7, no.~2, pp. 142--145, 2018.

\end{thebibliography}
\end{document}